\documentclass[11pt]{article}
\def\confversion{0}

\usepackage{ifthen}

\newcommand{\ignore}[1]{}

\ifthenelse{\equal{\confversion}{1}}
{
	\newcommand{\conf}[1]{#1}
}
{
	\newcommand{\conf}[1]{\ignore{#1}}	
}
\ifthenelse{\equal{\confversion}{0}}
{
	\newcommand{\full}[1]{#1}
}
{
	\newcommand{\full}[1]{\ignore{#1}}	
}

\full{\usepackage{fullpage}}
\usepackage{graphicx}
\usepackage[pdfencoding=auto]{hyperref}
\hypersetup{
	colorlinks,
	linkcolor={red!50!black},
	citecolor={blue!50!black},
	urlcolor={blue!80!black}
}
\usepackage{xspace}
\usepackage{enumerate}
\usepackage{enumitem}
\usepackage{xcolor}

\usepackage{array}
\usepackage{mathdots}
\full{\usepackage{amsthm,amssymb,amsmath}}
\usepackage{amssymb,amsmath}
\usepackage{mathtools}
\usepackage{algorithm}
\usepackage{algpseudocode}
\usepackage{mathrsfs}
\usepackage[nameinlink]{cleveref}

\full{
  \newtheorem{theorem}{Theorem}
  \theoremstyle{definition}
  \newtheorem{definition}{Definition}
  \newtheorem{lemma}{Lemma}
  
  \newtheorem{corollary}{Corollary}
    \newtheorem*{theorem*}{\bf Informal Theorem}
   \newtheorem{remark}{Remark}
   \newtheorem{proposition}{Proposition}
   
}
\newtheorem{fact}{Fact}
\conf{\usepackage[pdfencoding=auto]{hyperref}
}

\newcommand{\etal}{{\em et al}.~}
\newcommand{\cgk}{\textsf{CGK}\xspace}

\newcommand{\cF}{\mathscr{F}}
\newcommand{\cC}{\mathcal{C}}
\newcommand{\cI}{\mathcal{I}}

\newcommand{\cT}{\mathcal{T}}
\newcommand{\R}{\mathbb{R}}
\newcommand{\N}{\mathbb{N}}
\newcommand{\w}{\mathsf{w}}
\newcommand{\m}{\mathsf{m}}

\newcommand{\cov}{\mathsf{cov}}
\newcommand{\optval}{{\rho}}
\newcommand{\chld}{{\mathsf{Child}}}
\newcommand{\leaf}{{\mathsf{Leaf}}}
\newcommand{\HS}{\textsf{HS}\xspace}
\newcommand{\CovP}{\mathscr{P}^\cI_\cov}
\newcommand{\tree}{\text{$(L_1,L_2,\leaf,\w)$}\xspace}

\newcommand{\nukc}{\textsf{NU$k$C}\xspace}
\newcommand{\trkc}{\textsf{$2$-NU$k$C}\xspace}

\newcommand{\trkco}{\textsf{Robust $2$-NU$k$C}\xspace}
\newcommand{\trkcoinst}{\text{$((X,d),(r_1,r_2),(k_1,k_2),\m)$}\xspace}
\newcommand{\tlff}{\textsf{2-FF}\xspace}
\newcommand{\tlffinst}{\text{$(\tree,(k_1,k_2))$}\xspace}
\newcommand{\septrkco}{\textsf{Well-Separated Robust \trkc}\xspace}

\newcommand{\bx}{\mathbf{x}}
\newcommand{\by}{\mathbf{y}}
\begin{document}
%
\title{Robust $k$-Center with Two Types of Radii}
%
%
\author{Deeparnab Chakrabarty\thanks{Partially supported by NSF grant \#1813053} \and Maryam Negahbani}
%
\conf{\authorrunning{D. Chakrabarty and M. Negahbani}}
%
\conf{\institute{Dartmouth College, Hanover NH 03755, USA \\\email{deeparnab@dartmouth.edu,~} \email{maryam@cs.dartmouth.edu}}}

\date{}
\maketitle              
\begin{abstract}
In the non-uniform $k$-center problem, the objective is to cover points in a metric space with 
specified number of balls of different radii. Chakrabarty, Goyal, and Krishnaswamy [ICALP 2016, Trans.\ on Algs.\ 2020] (CGK, henceforth) give a
constant factor approximation when there are two types of radii. In this paper, we give a constant factor approximation
for the two radii case in the presence of outliers. To achieve this, we need to bypass the technical barrier of bad integrality gaps in the CGK approach. We do so using ``the ellipsoid method inside the ellipsoid method'': use an outer layer of the ellipsoid method to reduce to stylized instances and use an inner layer of the ellipsoid method to solve these specialized instances. This idea is of independent interest and could be applicable to other problems.

\conf{}
\end{abstract}
\thispagestyle{empty}
\noindent
\section{Introduction}\label{sec:intro}

In the non-uniform $k$-center (\nukc) problem, one is given a metric space $(X,d)$ and balls of different radii $r_1 > \cdots > r_t$, with $k_i$ balls of radius type $r_i$.
The objective is to find a placement $C\subseteq X$ of centers of these $\sum_i k_i$ balls, such that they cover $X$ with as little dilation as possible. More precisely,
for every point $x\in X$ there must exist a center $c\in C$ of some radius type $r_i$ such that $d(x,c) \leq \alpha\cdot r_i$ and the objective is to find $C$ with $\alpha$ as small as possible.

Chakrabarty, Goyal, and Krishnaswamy~\cite{CGK16} introduced this problem as a generalization to the vanilla $k$-center problem~\cite{Gon85,HS85,HS86} which one obtains with only one type of radius.
 One motivation arises from source location and vehicle routing: imagine you have a fleet of $t$-types of vehicles of different speeds and your objective is to find depot locations so that any client point can be served as fast as possible. This can be modeled as an \nukc problem. The second motivation arises in clustering data. The $k$-center objective forces one towards clustering with equal sized balls, while the \nukc objective gives a more nuanced way to model the problem. Indeed, \nukc generalizes the {\em robust} $k$-center problem~\cite{CKMN01} which allows the algorithm to throw away $z$ points as outliers.
 This is precisely the \nukc problem with two types of radii, $r_1 = 1$, $k_1 = k$, $r_2 = 0$, and $k_2 = z$.

Chakrabarty et al.~\cite{CGK16} give a $2$-approximation for the special case of robust $k$-center which is the best possible~\cite{HS85,Gon85}. Furthermore, they give a $(1+\sqrt{5})$-factor approximation algorithm for the \nukc problem with two types of radii (henceforth, the \trkc problem).~\cite{CGK16} also prove that when $t$, the number of types of radii, is part of the input, there is no constant factor approximation algorithms unless P=NP. They explicitly leave open the case when the number of different radii types is a constant, conjecturing that constant-factor approximations should be possible. We take the first step towards this by looking at the {\em robust} \trkc problem. That is, the \nukc problem with two kinds of radii when we can throw away $z$ outliers. This is the case of $3$-radii with $r_3 = 0$.

\begin{theorem}\label{thm:mainthm}
	There is a $10$-approximation for the {\em \trkco} problem.
\end{theorem}
\noindent
Although the above theorem seems a modest step towards the CGK conjecture, it is in fact a non-trivial one which bypasses multiple technical barriers in the~\cite{CGK16} approach.
To do so, our algorithm applies a two-layered round-or-cut framework, and it is foreseeable that this idea will form a key ingredient for the constantly many radii case as well.
In the rest of this section, we briefly describe the~\cite{CGK16} approach, the technical bottlenecks one faces to move beyond $2$ types of radii, and our approach to bypass them. A more detailed description appears in~\Cref{sec:prelims}.

One key observation of~\cite{CGK16} connects \nukc with the {\em firefighter problem on trees}~\cite{FKMR07,CC10,ABZ16}. In the latter problem, one is given a tree where there is a fire at the root.
The objective is to figure out if a specified number of firefighters can be placed in each layer of the tree, so that the leaves can be saved. To be precise, the objective is to select $k_i$ nodes from layer $i$ of the tree so that every leaf-to-root path contains at least one of these selected nodes. 

Chakrabarty et al.~\cite{CGK16}  use the {\em integrality} of a natural LP relaxation for the firefighter problem on height-$2$ trees to obtain their constant factor approximation for \trkc. In particular, they show how to convert a fractional solution of the standard LP relaxation of the \trkc problem to a feasible fractional solution for the firefighter LP. Since the latter LP is integral for height-2 trees, they obtain an integral firefighting solution from which they construct an $O(1)$-approximate solution for the \trkc problem. Unfortunately, this idea breaks down in the presence of outliers as the firefighter LP on height-2 trees {\em when certain leaves can be burnt} (outlier leaves, so to speak) is not integral anymore.
In fact, the standard LP-relaxation for \trkco has unbounded integrality gap. This is the first bottleneck in the CGK approach.

Although the LP relaxation for the firefighter problem on height-2 trees is not integral when some leaves can be burnt, the problem itself (in fact for any constant height) is solvable in polynomial time using dynamic programming (DP). Using the DP, one can then obtain~(see, for instance, \cite{Kaibel11}) a {\em polynomial sized integral} LP formulation for the firefighting problem. 
 This suggests the following enhancement of the CGK approach using the ellipsoid method. Given a fractional solution $\bx$ to \trkco, use the CGK approach to obtain a fractional solution $\by$ to the firefighting problem. If $\by$ is feasible for the integral LP formulation, then we get an integral solution to the firefighting problem which in turn gives an $O(1)$-approximation for the \trkco instance via the CGK approach. Otherwise, we would get a {\em separating hyperplane} for $\by$ and the poly-sized integral formulation for firefighting. 
If we could only use this to separate the fractional solution $\bx$ from the integer hull of the \trkco problem, then we could use the ellipsoid method to approximate \trkco. This is the so-called ``round-or-cut'' technique in approximation algorithms. 

Unfortunately, this method also fails and indicates a much more serious bottleneck in the CGK approach. Specifically, there is an instance of \trkco and an $\bx$ in the integer hull of its solutions, such that the firefighting instance output by the CGK has {\em no} integral solution! Thus, one needs to enhance the CGK approach in order to obtain $O(1)$-approximations even for the \trkco problem. The main contribution of this paper is to provide such an approach. We show that if the firefighting instance does not have an integral solution, then we can tease out many stylized \trkco instances on which the round-or-cut method provably succeeds, and an $O(1)$-approximation to any one of them gives an $O(1)$-approximation to the original \trkco instance. 

\medskip

\noindent
{\bf \em Our Approach.} Any solution $\bx$ in the integer hull of \nukc solutions gives an indication of where different radii centers are opened. As it turns out, the key factor towards obtaining algorithms for the \trkco problem is observing where the large radii (that is, radius $r_1$) balls are opened. Our first step is showing that if the fractional solution $\bx$ tends to open the $r_1$-centers only on ``well-separated'' locations then in fact, the round-or-cut approach described above works. More precisely, if the \trkco instance is for some reason forced to open its $r_1$ centers on points which are at least $cr_1$ apart from each other for some constant $c > 4$, then the CGK approach plus round-or-cut leads to an $O(1)$-approximation for the \trkco problem. 
We stress that this is far from trivial and the natural LP relaxations have bad gaps even in this case. We use our approach from a previous paper~\cite{CN18} to handle these well-separated instances.

But how and why would such  well-separated instances arise? This is where we use ideas from recent papers on fair colorful clustering~\cite{BIPV19,JSS20,AAKZ20}. If $\bx$ suggested that the $r_1$-radii centers are not well-separated, then one does not need that many balls if one allows dilation. In particular, if $p$ and $q$ are two $r_1$-centers of a feasible integral solution,  and $d(p,q) \leq cr_1$, then just opening one ball at either $p$ or $q$ with radius $(c+1)r_1$ would cover every point that they each cover with radius $r_1$-balls. 
Thus, in this case, the approximation algorithm gets a ``saving'' in the budget of how many balls it can open. We exploit this savings in the budget by utilizing yet another observation from Adjiashvili, Baggio, and Zenklusen~\cite{ABZ16} on the natural LP relaxation for the firefighter problem on trees. This asserts that although the natural LP relaxation for constant height trees is not integral, one can get integral solutions by violating the constraints {\em additively} by a constant. The aforementioned savings allow us to get a solution without violating the budget constraints.

In summary, given an instance of the \trkco problem, we run an outer round-or-cut framework and use it to check whether an instance is well-separated or not. If not, we straightaway get an approximate solution via the CGK approach and the ABZ observation. Otherwise, we use enumeration (similar to~\cite{AAKZ20}) to obtain $O(n)$ many different {\em well-separated} instances and for each, run an inner round-or-cut framework. If any of these well-separated instances are feasible, we get an approximate solution for the initial \trkco instance. Otherwise, we can assert a separating hyperplane for the outer round-or-cut framework. \medskip

\noindent
{\bf \em Related Work.} 
\nukc was introduced in~\cite{CGK16} as a generalization to the $k$-center problem~\cite{Gon85,HS85,HS86} and the robust $k$-center problem~\cite{CKMN01}. In particular CGK reduce \nukc to the firefighter problem on trees which has constant approximations~\cite{FKMR07,CC10,ABZ16} and recently, a quasi-PTAS~\cite{RS20}. \nukc has also been studied in the \emph{perturbation resilient}~\cite{ABS12,AMM17,CG18} settings. An instance is $\rho$-perturbation resilient if the optimal clustering does not change even when the metric is perturbed up to factor $\rho$.
Bandapadhyay~\cite{B20} gives an exact polynomial time algorithm for $2$-perturbation resilient instances with constant number of radii.

As mentioned above, part of our approach is inspired by ideas from fair colorful $k$-center clustering~\cite{BIPV19,JSS20,AAKZ20} problems studied recently. In this problem, the points are divided into $t$ color classes and we are asked to cover $m_i$, $i \in \{1,\dots,t\}$ many points from each color by opening $k$-centers. The idea of moving to well-separated instances are present in these papers. We should mention, however, that the problems are different, and their results do not imply ours.

The round-or-cut framework is a powerful approximation algorithm technique 	first used in a paper by Carr \etal~\cite{CFLP00} for the minimum knapsack problem, and since then has found use in other areas such as network design~\cite{CCKK15} and clustering~\cite{ASS17,L15,L16,CN18,AAKZ20}. Our multi-layered round-or-cut approach may find uses in other optimization problems as well.

\section{Detailed Description of Our Approach}\label{sec:prelims}

In this section, we provide the necessary technical preliminaries required for proving \Cref{thm:mainthm} and give a more detailed description of the CGK bottleneck and our approach. 
We start with  notations. Let $(X,d)$ be a metric space on a set of points $X$ with distance function $d:X \times X \longrightarrow \mathbb{R}_{\geq 0}$ satisfying the triangle inequality. 
For any $u\in X$ we let $B(u,r)$ denote the set of points in a ball of radius $r$ around $u$, that is, $B(u,r) = \{v \in X : d(u,v) \leq r\}$. For any set $U \subseteq X$ and function $f:U \rightarrow \R$, we use the shorthand notation $f(U) := \sum_{u \in U} f(u)$. For a set $U \subseteq X$ and any $v \in X$ we use $d(v,U)$ to denote $\min_{u \in U} d(v,u)$.

The $2$-radii \nukc problem and the robust version are formally defined as follows.

\begin{definition}[\trkc and \trkco]\label{def:trkc}\label{def:trkco}
The input to \trkc is a metric space $(X,d)$ along with two radii $r_1 > r_2 \geq 0$ with respective budgets $k_1,k_2 \in \N$. 
The objective of \trkc is to find the minimum $\optval \geq 1$ for which there exists subsets $S_1, S_2 \subseteq X$ such that (a) $|S_i| \leq k_i$ for $i \in \{1,2\}$, and (b) $\bigcup_i\bigcup_{u \in S_i} B(u,\optval r_i) = X$. 
The input to \trkco contains an extra parameter $\m \in \N$, and the objective is the same, except that condition (b) is changed to $|\bigcup_i\bigcup_{u \in S_i} B(u,\optval r_i)| \geq \m$.
\end{definition}
\noindent
An instance $\cI$ of \trkco is denoted as \trkcoinst. As is standard, we will focus on the {\em approximate feasibility} version of the problem. An algorithm for this problem takes input an instance $\cI$ of \trkco, and either asserts that $\cI$ is {\em infeasible}, that is, there is no solution with $\rho = 1$, or provides a solution with $\rho \leq \alpha$. 
Using binary search, such an algorithm implies an $\alpha$-approximation for \trkco.

{\em Linear Programming Relaxations.}
The following is the natural LP relaxation for the feasibility version of \trkco. For every point $v\in X$, $\cov_i(v)$ denotes its {\em coverage} by balls of radius $r_i$.
Variable $x_{i,u}$ denotes the extent to which a ball of radius $r_i$ is open at point $u$. If instance $\cI$ is feasible, then the following polynomial sized system of inequalities has a feasible solution.
\begin{align}
   \{(\cov_i(v): v\in X, i \in \{1,2\}) :~~~ \sum_{v \in X} \cov(v)& \geq \m & \label{lp:trkco}\tag{\trkco LP}\\
    \sum_{u \in X} x_{i,u} &\leq k_i & \forall i \in \{1,2\} \notag\\
    \cov_1(v) = \sum_{u \in B(v,r_1)} x_{1,u},~~  \cov_2(v) = \sum_{u \in B(v,r_2)} &x_{2,u}  &\forall v \in X \notag\\
    \cov(v) = \cov_1(v) + \cov_2(v) &\leq 1   & \forall v \in X \notag\\
    x_{i,u} &\geq 0                     & \forall i \in \{1,2\}, \forall u \in X\notag\}
\end{align}
For our algorithm, we will work with the following integer hull of all possible fractional coverages.
Fix a \trkco instance $\cI = \trkcoinst$ and let $\cF$ be the set of all tuples of subsets $(S_1,S_2)$ with $|S_i|\leq k_i$.
For $v \in X$ and $i \in \{1,2\}$, we say $\cF$ covers $v$ with radius $r_i$ if $d(v,S_i) \leq r_i$. 
Let $\cF_i(v) \subseteq \cF$ be the subset of solutions that cover $v$ with radius $r_i$. 
Moreover, we would like $\cF_1(v)$ and $\cF_2(v)$ to be disjoint, so if $S \in \cF_1(v)$, we do not include it in $\cF_2(v)$.
The following is the integer hull of the coverages. If $\cI$ is feasible, there must exist a solution in $\CovP$.

\begin{alignat}{4}
 \{(\cov_i(v): v\in X, i \in \{1,2\}) : 
 && \sum_{v \in X}  \left(\cov_1(v) + \cov_2(v)\right) & \geq & ~~\m \label{eq:P1} \tag{$\CovP$}  \\
\forall v\in X,i \in \{1,2\} && ~~\cov_i(v) - \sum\limits_{S \in \cF_i(v)} z_S &=& ~~0 \tag{$\CovP$.1} \label{eq:P2} \\
&& 	\sum\limits_{S \in \cF} z_S & = & ~~ 1  \tag{$\CovP$.2} \label{eq:P3} \\
\forall S\in \cF  && z_S &\geq & 0\} \tag{$\CovP$.3} \label{eq:P4}
\end{alignat}
\begin{fact}\label{clm:configiInLP}
	$\CovP$ lies inside \ref{lp:trkco}.
\end{fact}
\noindent
{\bf \em Firefighting on Trees.} As described in~\Cref{sec:intro}, the CGK approach~\cite{CGK16} is via the firefighter problem on trees. 
Since we only focus on \trkco, the relevant problem is the {\em weighted} $2$-level fire fighter problem. The input includes a set of height-2 trees (stars) with root nodes $L_1$ and leaf nodes $L_2$. Each leaf $v \in L_2$ has a parent $p(v) \in L_1$ and an integer weight $\w(v) \in \N$. We use $\leaf(u)$ to denote the leaves connected to a 
$u \in L_1$ (that is, $\{v \in L_2: p(v) = u\}$). Observe that $\{\leaf(u) : u \in L_1\}$ partitions $L_2$. So we could represent the edges of the trees by this $\leaf$ partition. Hence the structure is identified as $(L_1,L_2,\leaf,\w)$.
\begin{definition}[2-Level Fire Fighter (\tlff) Problem]\label{def:tlff}
Given height-2 trees $\tree$ along with budgets $k_1,k_2 \in \N$, a feasible solution is a pair $T=(T_1,T_2)$, $T_i \subseteq L_i$, such that $|T_i| \leq k_i$ for $i \in \{1,2\}$. Let $\cC(T) := \{v \in L_2: v \in T_2 \lor p(v) \in T_1\}$ be the set of leaves covered by $T$. The objective is to maximize $\w(\cC(T))$.  Hence a \tlff instance is represented by $((L_1,L_2,\leaf,\w),k_1,k_2)$.
\end{definition}
\noindent
The standard LP relaxation for this problem is quite similar to the \ref{lp:trkco}. For each vertex $u \in L_1\cup L_2$ there is a variable $0 \leq y_u \leq 1$ that shows the extent to which $u$ is included in the solution. For a leaf $v$, $Y(v)$ is the fractional amount by which $v$ is covered through both itself and its parent. 
\begin{align}
\max\sum_{v \in L_2}\w(v)Y(v):~~
&\sum_{u \in L_i} y_u \leq k_i,~~ \forall i \in \{1,2\}; ~~~\label{lp:tlff}\tag{\tlff LP}\\
&Y(v) := y_{p(v)} + y_v \leq 1, ~\forall v \in L_2;~~~
y_u \geq 0,~~ \forall u \in L_1\cup L_2 \notag
\end{align}
\remark{The following figure shows an example where the above LP relaxation has an integrality gap. However,
	\tlff can be solved via dynamic programming in $O(n^3)$ time and has similar sized integral LP relaxations.}
\begin{figure}[!ht]
	\centering
	\includegraphics[scale = 0.5]{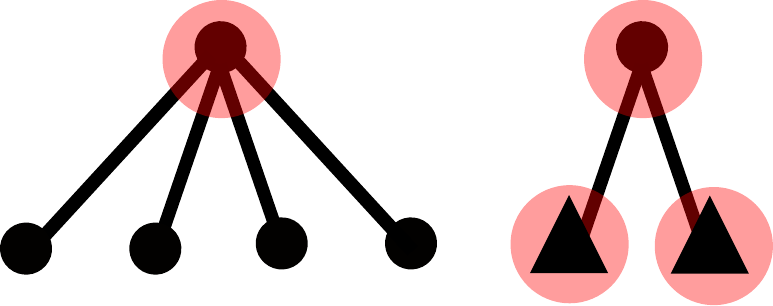}
	\caption{A \tlff instance with budgets $k_1 = k_2 = 1$. Multiplicity $\w$ is 1 for the circle leaves and 3 for the triangles. The highlighted nodes have $y = 1/2$ and the rest of the nodes have $y=0$. 
		The objective value for this $y$ is $4 \times 1/2 + 6 = 8$ but no integral solution can get an objective value of more than 7.}
	\label{fig:gaplvl2}
\end{figure}
\subsection{\cgk's Approach and its Shortcomings}
Given fractional coverages $(\cov_1(v), \cov_2(v)~:~v\in X)$, the CGK algorithm~\cite{CGK16} runs the classic clustering subroutine by Hochbaum and Shmoys~\cite{HS85}
in a greedy fashion. In English, the Hochbaum-Shmoys (HS) routine partitions a metric space such that the representatives of each part are well-separated with respect to an input parameter. The CGK algorithm obtains a \tlff instance by applying the HS routine twice. Once on the whole metric space in decreasing order of $\cov(v) = \cov_1(v) + \cov_2(v)$, and the set of representatives forms the leaf layer $L_2$ with weights being the size of the parts. The next time on $L_2$ itself in decreasing order of $\cov_1$ and the representatives form the parent layer $L_1$. These subroutines and the subsequent facts form a part of our algorithm and analysis.
\begin{algorithm}[ht]
	\caption{\HS}
	\label{alg:HS}
	\begin{algorithmic}[1]
		\Require Metric $(U,d)$, parameter $r \geq 0$, and assignment $\{\cov(v) \in \R_{\geq 0} :v\in  U\}$
		\State $R \leftarrow \emptyset$ \Comment{The set of representatives}
		\While{ $U \neq \emptyset$} 
	        \State $u \leftarrow \arg\max_{v\in U} \cov(v)$ \Comment{The first client in $U$ in non-increasing $\cov$ order} \label{ln:greedy}
	        \State $R \leftarrow R \cup u$ \label{ln:rep}
	        \State $\chld(u) \leftarrow \{v \in U: d(u,v) \leq r\}$\Comment{Points in $U$ at distance $r$ from $u$ (including $u$ itself)} \label{ln:chld}
	        \State $U \leftarrow U \backslash \chld(u)$\label{ln:remove-from-U}
		\EndWhile
		\Ensure $R$, $\{\chld(u) : u \in R\}$
	\end{algorithmic}
\end{algorithm}
\begin{algorithm}[ht]
		\caption{\cgk}
		\label{constAlgo}\label{alg:reduce}
		\begin{algorithmic}[1]
			\Require \trkco instance \trkcoinst, dilation factors $\alpha_1, \alpha_2 > 0$, and assignments $\cov_1(v), \cov_2(v) \in \R_{\geq 0}$ for all $v\in X$ 
			\State $(L_2, \{\chld_2(v), v \in L_2\}) \leftarrow \HS((X,d),\alpha_2r_2,\cov = \cov_1 + \cov_2)$ \label{ln:L2}
		    \State $(L_1, \{\chld_1(v), v \in L_1\}) \leftarrow \HS((L_2,d),\alpha_1r_1,\cov_1)$\label{ln:L1}
		    \State $\w(v) \leftarrow |\chld_2(v)|$ for all $v \in L_2$ \label{ln:w}
		    \State $\leaf(u) \leftarrow \chld_1(u)$ for all $u \in L_1$\label{ln:leaf}
            \Ensure \tlff instance \tlffinst
		\end{algorithmic}
\end{algorithm}
\begin{definition}[Valuable \tlff instances]
	We call an instance $\cT$ returned by the \cgk algorithm valuable if it has an {\em integral} solution of total weight at least $\m$.
	Using dynamic programming, there is a polynomial time algorithm to check whether $\cT$ is valuable.
\end{definition}
\begin{fact}\label{fact:HS}
	The following are true regarding the output of \HS:
		(a) $\forall u \in R, \forall v \in \chld(u): d(u,v) \leq r$, (b)
	$\forall u,v \in R: d(u,v) > r$, (c)
	The set $\{\chld(u) : u \in R\}$ partitions $U$, and (d)
$\forall{u \in R},\forall{v \in \chld(u)}: \cov(u) \geq \cov(v)$.	
\end{fact}
\begin{lemma}[rewording of Lemma 3.4. in \cite{CGK16}]\label{lma:approx}
Let $\cI$ be a \trkco instance. If for any fractional coverages $(\cov_1(v),\cov_2(v))$ the instance \tlff created by \Cref{alg:reduce} is valuable, then one obtains an $(\alpha_1 + \alpha_2)$-approximation for $\cI$.
\end{lemma}
\noindent
\Cref{lma:approx} suggests that if we can find fractional coverages so that the corresponding \tlff instance $\cT$ is valuable, then we are done. 
Unfortunately, the example illustrated in \Cref{fig:gaplvl2legitmetric} 
shows that for any $(\alpha_1, \alpha_2)$ there exists \trkco instances and fractional coverages $(\cov_1(v),\cov_2(v)) \in \CovP$
in the integer hull, for which the CGK algorithm returns \tlff instances that are not valuable.

	\begin{figure}[ht]
		\centering
		\includegraphics[width=0.7\textwidth]{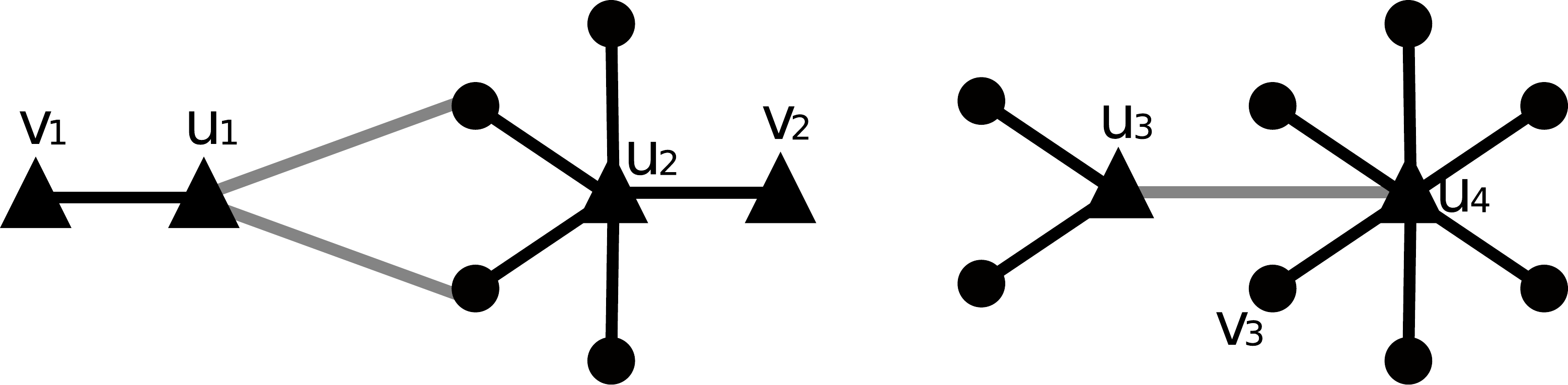}
		\includegraphics[width=0.7\textwidth]{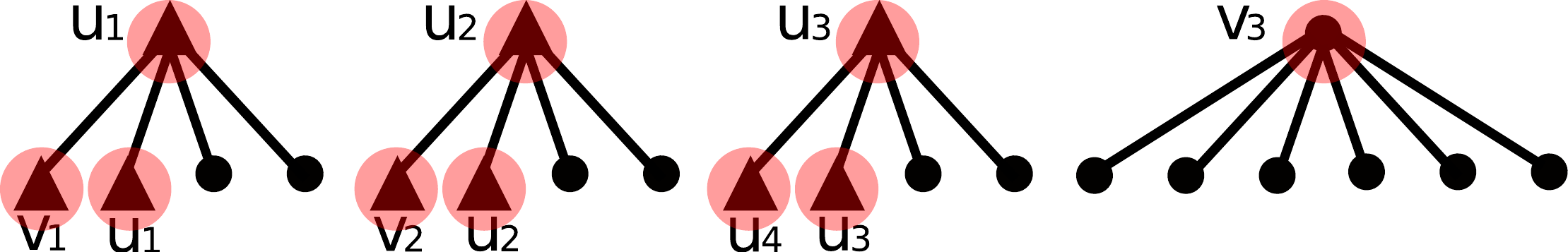}
		\caption{At the top, there is a feasible \trkco instance with $k_1 = 2$, $k_2 = 3$, and $\m = 24$. There are 6 triangles representing 3 collocated points each, along with 12 circles, each representing one point. The black edges are distance $r_1 > \alpha_2r_2$ and the grey edges are distance $\alpha_1r_1$. There are two integral solutions $S$ and $S'$ each covering exactly 24 points. $S_1 = \{u_1,u_4\}$, $S_2 = \{u_2,v_2,u_3\}$, $S'_1 = \{u_2,u_3\}$, and $S'_2 = \{u_1,v_1,u_4\}$. Having $z_S = z_{S'} = 1/2$ in $\CovP$, gives $\cov_1$ of $1/2$ for all the points and $\cov_2$ of $1/2$ for the triangles. The output of \Cref{alg:reduce} is the \tlff instance at the bottom. According to \Cref{clm:givesFracSln} the highlighted nodes have $y = 1/2$ and the rest of the nodes have $y=0$ with objective value $12 \times 1/2 + 18 = 24$ but no integral solution can get an objective value of more than 23.}
		\label{fig:gaplvl2legitmetric}
	\end{figure}
\subsection{Our Idea}\label{sec:ouridea}
Although the \tlff instance obtained by \Cref{alg:reduce} from fractional coverages $(\cov_1(v), \cov_2(v):v\in X)$ may not be valuable,~\cite{CGK16} proved that if these coverages come from \eqref{lp:trkco}, then there is always a {\em fractional} solution to \eqref{lp:tlff} for this instance which has value at least$\m$.
\begin{proposition}[rewording of Lemma 3.1. in \cite{CGK16}]\label{clm:givesFracSln}
	Let $(\cov_1(v), \cov_2(v): v\in X)$ be any feasible solution to \ref{lp:trkco}.  As long as $\alpha_1,\alpha_2 \geq 2$, the following is a fractional solution of \ref{lp:tlff} with value at least $\m$ for the \tlff instance output by \Cref{alg:reduce}.
	\begin{equation*}
	y_v = \begin{cases} \cov_1(v) & v \in L_1\\
	\min\{\cov_2(v), 1-\cov_1(p(v))\} & v \in L_2.
	\end{cases}
	\end{equation*}
\end{proposition}
\noindent
Therefore, the problematic instances are precisely \tlff instances that are integrality gap examples for \eqref{lp:tlff}. Our first observation stems from what Adjiashvili, Baggio, and Zenklusen~\cite{ABZ16} call ``the narrow integrality gap of the firefighter LP''.

\begin{lemma}[From Lemma 6 of \cite{ABZ16}]\label{lma:loose}
	Any basic feasible solution $\{y_v : i \in \{1,2\}, v \in L_i\}$ of the \ref{lp:tlff} polytope has at most 2 \emph{loose} variables. A variable $y_v$ is loose if $0 <y_v<1$ and $y_{p(v)} = 0$ in case $v \in L_2$.
\end{lemma}
\noindent
In particular, if $y(L_1) \leq k_1 - 2$, then the above lemma along with \Cref{clm:givesFracSln} implies there exists an {\em integral} solution with value $\geq m$. That is, the \tlff instance is valuable. Conversely, the fact that the instance is {\em not} valuable asserts that
$y(L_1) > k_1 - 2$ which in turn implies $\cov_1(L_1) > k_1 - 2$. In English, the fractional coverage puts a lot of weight on the points in $L_1$.

This is where we exploit the ideas in~\cite{BIPV19,JSS20,AAKZ20}. By choosing $\alpha_1 > 2$ to be large enough  in \Cref{clm:givesFracSln}, we can ensure that points in $L_1$ are ``well-separated''. More precisely,  we can ensure for any two $u,v\in L_1$ we have $d(u,v) > \alpha_1r_1$ (from \Cref{fact:HS}). The well-separated condition implies that the same center cannot be fractionally covering two different points in $L_1$. Therefore, $\cov_1(L_1) > k_1 - 2$ if $(\cov_1,\cov_2) \in \CovP$ is in the integer hull, then there must exist an integer solution which opens {\em at most} $1$ center that does \emph{not} cover points in $L_1$. For the time being assume in fact no such center exists and $\cov_1(L_1) = k_1$. Indeed, the integrality gap example in \Cref{fig:gaplvl2legitmetric} satisfies this equality. 

Our last piece of the puzzle is that if the $\cov_1$'s are concentrated on separated points, then indeed we can apply the round-or-cut framework to obtain an approximation algorithm. To this end, we make the following definition, and assert the following theorem.
\begin{definition}[\septrkco]
The input is the same as \trkco, along with $Y\subseteq X$ where $d(u,v) > 4r_1$ for all pairs $u,v \in Y$, and the algorithm is allowed to open the radius $r_1$-centers {\em only} on points in $Y$.
\end{definition}
\begin{theorem}\label{thm:approxsep}
	Given a \septrkco instance there is a polynomial time algorithm using the ellipsoid method that either gives a $4$-approximate solution, or proves that the instance is infeasible.
\end{theorem}
\noindent
We remark the natural \eqref{lp:trkco} relaxation still has a bad integrality gap, and we need the round-or-cut approach. Formally, given fractional coverages $(\cov_1,\cov_2)$ we run \Cref{alg:reduce} (with $\alpha_1 = \alpha_2 = 2$) to get a \tlff instance.
If the instance is valuable, we are done by \Cref{lma:approx}. Otherwise, we prove that $(\cov_1,\cov_2) \notin \CovP$ by exhibiting a separating hyperplane. This crucially uses the well-separated-ness of the instance and indeed, the bad example shown in~\Cref{fig:gaplvl2legitmetric} is {\em not} well-separated. This implies \Cref{thm:approxsep} using the ellipsoid method. \smallskip

In summary, to prove~\Cref{thm:mainthm}, we start with $(\cov_1,\cov_2)$ purported to be in $\CovP$. Our goal is to either get a constant approximation, or separate $(\cov_1,\cov_2)$ from $\CovP$.
We first run the CGK~\Cref{alg:reduce} with $\alpha_1 = 8$ and $\alpha_2 = 2$. If $\cov_1(L_1) \leq k_1 - 2$, we can assert that the \tlff instance is valuable and get a $10$-approximation. Otherwise, $\cov_1(L_1) > k_1 - 2$, and we guess the $O(n)$ many possible centers ``far away'' from $L_1$, and obtain that many well-separated instances. We run the algorithm promised by ~\Cref{thm:approxsep} on each of them. If any one of them gives a $4$-approximate solution, then we immediately get an $8$-approximate\footnote{The factor doubles as we need to double the radius, but that is a technicality.} solution to the original instance. If {\em all} of them fail, then we can assert $\cov_1(L_1) \leq k_1 -2$ must be a {\em valid} inequality for $\CovP$, and thus obtain a hyperplane separating $(\cov_1,\cov_2)$ from $\CovP$. The polynomial running time is implied by the ellipsoid algorithm. Note that there are two nested runs of the ellipsoid method in the algorithm. \Cref{fig:flow} below shows an illustration of the ideas.

\begin{figure}[ht]
    \centering
    \includegraphics[width=\textwidth]{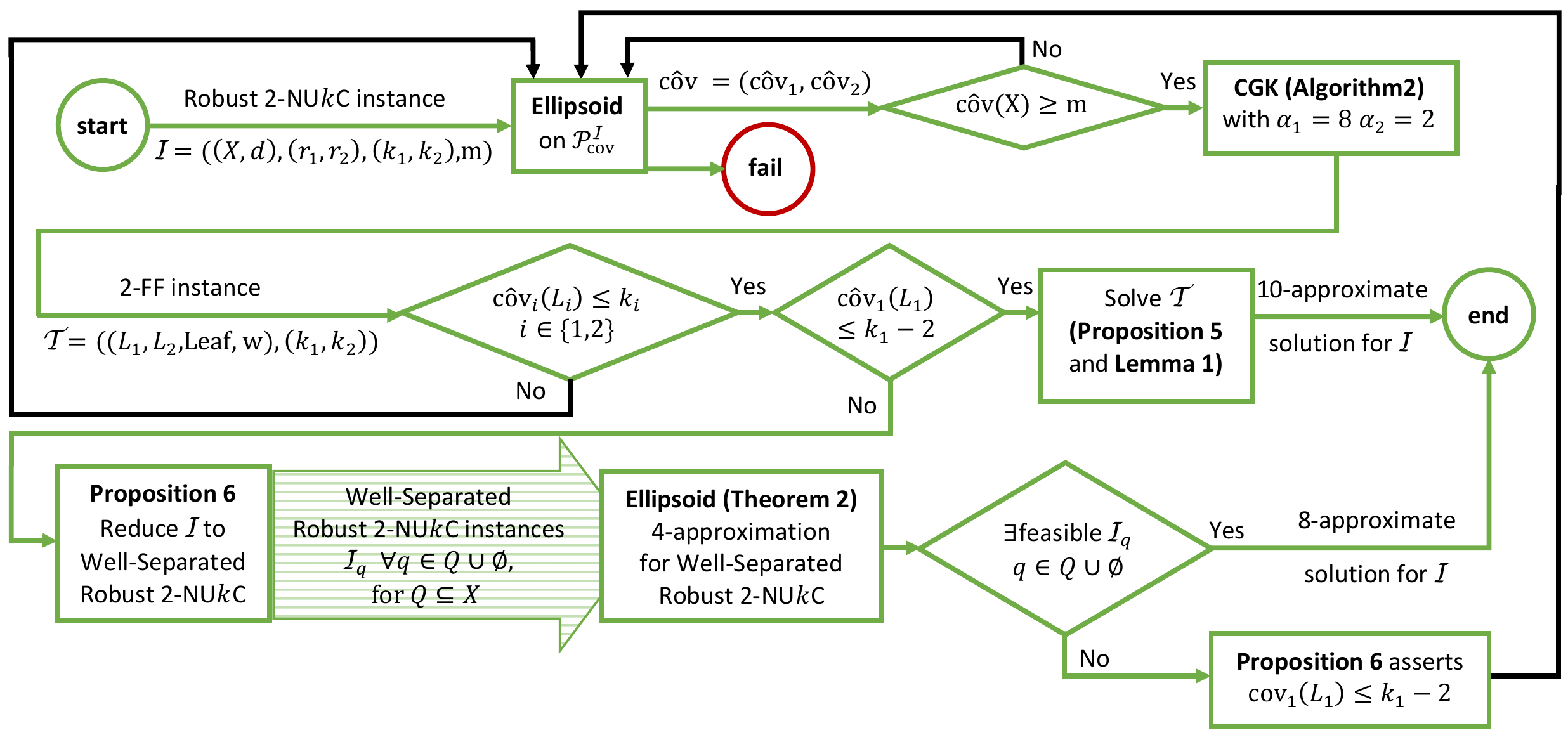}
    \caption{Our framework for approximating \trkco. The three black arrows each represent separating hyperplanes we feed to the outer ellipsoid. The box in the bottom row stating ``4-approximation for well-separated \trkco'' runs the 
    	inner ellipsoid method.}
    \label{fig:flow}
\end{figure}

\subsection{Discussion}
Before we move to describing algorithms proving~\Cref{thm:approxsep} and~\Cref{thm:mainthm}, let us point out why the above set of ideas does {\em not} suffice to prove the full CGK conjecture, that is, give an $O(1)$-approximation for \nukc with constant many type of radii. Given fractional coverages, the CGK algorithm now returns a $t$-layered firefighter instance and again if such an instance is valuable (which can be checked in $n^{O(t)}$ time), we get an $O(1)$-approximation. As above, the main challenge is when the firefighter instance is not valuable.
~\Cref{thm:approxsep}, in fact, does generalize if {\em all} layers are separated. Formally, if there are $t$ types of radii, and there are $t$ sets $Y_1,\ldots, Y_t$ such that (a) any two points $p,q\in Y_i$ are well-separated, that is, $d(p,q) > 4r_i$, and (b) the $r_i$-radii centers are only allowed to be opened in $Y_i$, then in fact there is an $O(1)$-approximation for such instances. Furthermore, if we had fractional coverages  $(\cov_1,\cov_2,\ldots,\cov_t)$ such that in the $t$-layered firefighter instance returned, {\em all} layers have ``slack'', that is $\cov_i(L_i) \leq k_i - t$, then one can repeatedly use~\Cref{lma:loose} to show that the tree instance is indeed valuable.

The issue we do not know how to circumvent is when some layers have slack and some layers do not. In particular, even with $3$ kinds of radii, we do not know how to handle the case when the first layer $L_1$ is well-separated and $\cov_1(L_1) = k_1$, but the second layer has slack $\cov_2(L_2) \leq k_2 - 3$.~\Cref{lma:loose} does not help since all the loose vertices may be in $L_1$, but they cannot all be picked without violating the budget. At the same time, we do not know how to separate such $\cov$'s, or whether such a situation arises when $\cov$'s are in the integer hull.
We believe one needs more ideas to resolve the CGK conjecture.

\section{Approximating \septrkco}
In this section we prove~\Cref{thm:approxsep} stated in~\Cref{sec:ouridea}. As mentioned there, the main idea is to run the round-or-cut method, and in particular use ideas from a previous paper~\cite{CN18} of ours.
The main technical lemma is the following.
\def\pcov{\hat{\cov}}
\begin{lemma}\label{lma:sepsep}
	Given \septrkco instance $\cI$ and fractional coverages $(\pcov_1(v),\pcov_2(v))$, if the output of the CGK~\Cref{alg:reduce} is not valuable, there is a hyperplane separating $(\pcov_1(v),\pcov_2(v))$ from $\CovP$.
	Furthermore, the coefficients of this hyperplane are bounded in value by $|X|$.
\end{lemma}
\remark{We need to be careful in one place. Recall that \HS is used in the CGK~\Cref{alg:reduce}. We need to assert in \HS, that points $u$ with $d(u,Y) \leq r_1$ are prioritized over points $v$ with $d(v,Y) > r_1$ to be taken in $L_1$. This is w.l.o.g. since $\cov_1(v) = 0$ if $d(v,Y) > r_1$ by definition of \septrkco.}

	\noindent
Using the ellipsoid method, the above lemma implies~\Cref{thm:approxsep}\conf{ (See the full version of the paper for a detailed proof)}.
\full{
\begin{proof}[Proof of \Cref{thm:approxsep}] The goal is to either prove $\CovP$ is empty, or give a $4$-approximate solution. To do so, we run the ellipsoid algorithm. Each time the ellipsoid algorithm provides a purported fractional point
	$(\pcov_1(v), \pcov_2(v))\in \CovP$ and asks for a separating hyperplane. Given such a solution, we first check if $\pcov_1(v) = 0$ for all $v$ with $d(v,Y) > r_1$. By the well-separatedness property of $\cI$, this must be a valid equality and we can force the ellipsoid method to run over these equalities.
	Then we run CGK~\Cref{alg:reduce} with this $(\pcov_1,\pcov_2)$ and $\alpha_1 = \alpha_2 = 2$. If the resulting \tlff instance is valuable, we get a $4$-approximate solution by \Cref{lma:approx}. If not, \Cref{lma:sepsep} provides a separating hyperplane to feed to ellipsoid. Since our hyperplanes can be described in polynomial size, ellipsoid terminates in polynomial time, either giving us some $\pcov$ leading  to a $4$-approximation along the way, or prompts that $\CovP$ is empty thereby proving $\cI$ is infeasible.
\end{proof}
}
\noindent
The rest of this section is dedicated to proving \Cref{lma:sepsep}.
Fix a well-separated \trkco instance $\cI$. Recall that $Y\subseteq X$ is a subset of points, and the radius $r_1$ centers are only allowed to be opened at $Y$.
Let $\cT$ be the \tlff instance output by \Cref{alg:reduce} on $\cI$ and $\cov$ with $\alpha_1=\alpha_2=2$. Recall, $\cT = (\tree, k_1,k_2)$.
The key part of the proof is the following valid inequality in case $\cT$ is not valuable.
\begin{lemma}\label{lma:hyp:covchld}
If $\cT$ is not valuable  
 $\sum_{v \in L_2} \w(v)\cov(v) \leq \m-1$
for any $\cov(v)\in \CovP$.
\end{lemma}
\noindent
Before we prove~\Cref{lma:hyp:covchld}, let us show how it proves~\Cref{lma:sepsep}. Given $(\pcov_1,\pcov_2)$ we first check\footnote{recall, $\pcov(v) = \pcov_1(v) + \pcov_2(v)$.} that $\sum_{u\in X} \pcov(u) \geq \m$, or otherwise that would be the hyperplane separating it from $\CovP$.
Now recall that in \Cref{alg:reduce}, for $v \in L_2$, $\w(v) = |\chld_2(v)|$ which is the number of points assigned to $v$ by \HS (see \Cref{ln:L2} of \Cref{alg:reduce}). By definition of $\w$ and then parts d) and c) of \Cref{fact:HS},
\begin{equation*}
  \sum_{v \in L_2} \w(v)\pcov(v) = \sum_{v \in L_2} \sum_{u \in \chld(v)}\pcov(v) \geq \sum_{v \in L_2} \sum_{u \in \chld(v)}\pcov(u) = \sum_{u \in X} \pcov(u) \geq \m.  
\end{equation*}
\noindent
That is, $(\pcov_1,\pcov_2)$ violates the valid inequality asserted in~\Cref{lma:hyp:covchld}, and this would complete the proof of~\Cref{lma:sepsep}.
All that remains is to prove the valid inequality lemma above.

\begin{proof}[of~\Cref{lma:hyp:covchld}]
Fix a solution $\cov\in \CovP$ and note that this is a convex combination of coverages induced by {integral} feasible solutions in $\cF$. The main idea of the proof is to use the solutions in $\cF$ to construct {\em solutions} to the tree instance $\cT$. Since $\cT$ is {\em not} valuable, each of these solutions will have ``small'' value, and then we use this to prove the lemma. To this end, fix $S = (S_1,S_2) \in \cF$ where $|S_i| \leq k_i$ for $i\in \{1,2\}$. The corresponding solution $T = (T_1,T_2)$ for $\cT$ is defined as follows: For $i \in \{1,2\}$ and any $u \in L_i$, $u$ is in $T_i$ iff $S_i \in \cF_i(u)$. That is, $d(u,S_i)\leq r_i$.

\begin{proposition}\label{clm:Tinbudget}
$T$ satisfies the budget constraints $|T_i| \leq k_i$ for $i \in \{1,2\}$.
\end{proposition}
\full{\begin{proof}
For $i \in \{1,2\}$ and two different $u,v \in T_i$ by \Cref{fact:HS} and our choice of $\alpha_i = 2$, $d(u,v) > 2r_i$. By the triangle inequality, a facility in $S_i$ cannot cover both $u$ and $v$ meaning $|T_i| \leq |S_i| \leq k_i$.
\end{proof}}
\noindent
The next claim is the only place where we need the well-separated-ness of $\cI$.
Basically, we will argue that the leaves covered by $T_1$ capture all the points covered by $S_1$.
\begin{proposition}\label{clm:nonfrac}
If $u \in L_1$ but $u \notin T_1$ then no $v \in \leaf(u)$ can be covered by a ball of radius $r_1$ in $S_1$.
\end{proposition}
\full{\begin{proof}
We will prove the contrapositive by showing that if $u = p(v)$ and $v$ is covered through $f \in S_1$, then $u$ as well must be covered by the same $f$ and therefore, $u \in T_1$.
Consider the following two cases: either $d(u,Y) > r_1$ in which case, by our assumption on \HS, $v$ is prioritized over $u$ to be chosen in $L_1$ so this cannot happen. Thus, we must have $d(u,Y) \leq r_1$ which means there is $f_u \in Y$ with $d(u,f_u) \leq r_1$. This $f_u$ has to be equal to $f$ otherwise, by definition of $Y$ we must have $d(f,f_u) > 4r_1$ that contradicts the following:
   $d(f_u,f) \leq d(f_u, u) + d(u,v) + d(v,f) \leq r_1 + \alpha_1r_1 + r_1 = 4r_1$. \conf{\qed}
\end{proof}
}
Next, we can prove that overall, the leaves covered by $T$ capture the whole set of points covered by $S$. Recall that $\cC(T) = \{v \in L_2: v \in T_2 \lor p(v) \in T_1\}$ is the set of leaves covered by $T$. For $v \in X$ let $\cF(v) := \cF_1(v) \cup \cF_2(v)$ be the set of solutions that cover $v$.

\begin{proposition}\label{clm:lambdaleaf}
Take \tlff solution $T$ corresponding to \septrkco solution $S$ as described earlier.  We have:
\begin{equation*}
\sum\limits_{ \substack{v \in L_2: S \in \cF(v)}} \w(v) \leq \w(\cC(T)). 
\end{equation*}
That is, the total $\w$ of the points covered by $S$ is at most $\w(\cC(T))$.
\end{proposition}
\full{
\begin{proof} 
The leaves covered by $T$ are covered either by $T_1$ or $T_2$. Thus, we get
\begin{equation}\label{eq:007}
\w(\cC(T)) = \sum_{u \in T_1}\sum_{v \in \leaf(u)} \w(v) + \sum_{u \notin T_1}\sum_{\substack{v \in \leaf(u):v \in T_2}}\w(v).
\end{equation}
The first of these terms can be lower-bounded as
\[
\sum_{u \in T_1}\sum_{v \in \leaf(u)} \w(v) \geq \sum_{u \in T_1}\sum_{\substack{v \in \leaf(u): S \in \cF(v)}} \w(v). 
\]
that is, we only consider the leaves $v$ of $u\in T_1$ which are covered by the \trkco solution $S$.
The second term, by definition of $T_2$ is
\[
\sum_{u \notin T_1}\sum_{\substack{v \in \leaf(u):v \in T_2}}\w(v) = \sum_{u \notin T_1}\sum_{\substack{v \in \leaf(u): S \in \cF_2(v)}}\w(v) = \sum_{u \notin T_1}\sum_{\substack{v \in \leaf(u):S \in \cF(v)}}\w(v).
\]
where the last equality uses~\Cref{clm:nonfrac} which implies for $u\notin T_1$ and $v\in \leaf(u)$, $d(v,S_1) > r_1$. Thus, the solution $S$ covers $v$ iff $S_2$ covers $v$.
Plugging back in \eqref{eq:007}, we complete the proof.
\end{proof}
}
\noindent
The proof of~\Cref{lma:hyp:covchld} now follows from the fact that $\cT$ is not valuable thus $\w(\cC(T)) \leq \m-1$ and therefore, for any $S\in \cF$ we have $\sum_{v\in L_2: S\in \cF(v)} \w(v) \leq \m-1$. So we have:
\begin{align*}
\sum\limits_{v \in L_2} \w(v)\cov(v) &~~=_{\eqref{eq:P2}} ~~~ \sum\limits_{v \in L_2} \w(v)\sum\limits_{S \in \cF(v)} z_S  ~~= \sum\limits_{S \in \cF} z_S\sum\limits_{ \substack{v \in L_2:\\ S \in \cF(v)}}\w(v)\\
&~~\leq ~~~ (\m-1) \sum\limits_{S \in \cF} z_S =_{\eqref{eq:P3}} ~~~\m-1. \conf{\tag*{\qed}}
\end{align*} 
\end{proof}
\conf{Proofs of \Cref{clm:Tinbudget,clm:nonfrac,clm:lambdaleaf} can be found in the full version of the paper.}
\section{The Main Algorithm: Proof of~\Cref{thm:mainthm}}
As mentioned in~\Cref{sec:prelims}, we focus on the feasibility version of the problem: given an instance $\cI$ of \trkco we either want to prove it is infeasible, that is, there are no subsets $S_1,S_2\subseteq X$ with (a) $|S_i|\leq k_i$ and (b) $|\bigcup_i\bigcup_{u \in S_i} B(u, r_i)| \geq \m$, or give a $10$-approximation that is, open subsets $S_1,S_2$ that satisfy (a) and $|\bigcup_i\bigcup_{u \in S_i} B(u,10r_i)| \geq \m$. To this end, we apply the round-or-cut methodology on $\CovP$.
Given a purported $\pcov := (\pcov_1(v), \pcov_2(v): v\in X)$ we want to either use it to get a $10$-approximate solution, or find a hyperplane separating it from $\CovP$. Furthermore, we want the coefficients in the hyperplane to be poly-bounded.
Using the ellipsoid method we indeed get a polynomial time algorithm thereby proving~\Cref{thm:mainthm}.

Upon receiving $\pcov$, we first check whether $\pcov(X) \geq \m$ or not, and if not that will be the separating hyperplane. Henceforth, we assume this holds.
Then, we run CGK~\Cref{alg:reduce} with $\alpha_1 = 8$ and $\alpha_2 = 2$ to get \tlff instance $\cT = \tlffinst$. 
Let $\{y_v: v\in L_1\cup L_2\}$ be the solution described in~\Cref{clm:givesFracSln}.  Next, we check if $\pcov_i(L_i) = y(L_i) \leq k_i$ for both $i\in \{1,2\}$; if not, by~\Cref{clm:givesFracSln} that hyperplane would separate $\pcov$ from $\CovP$ (and even \ref{lp:trkco} in fact). The algorithm then branches into two cases.

\noindent
\textbf{Case I: $y(L_1) \leq k_1 - 2$.}
In this case, we assert that $\cT$ is valuable, and therefore by~\Cref{lma:approx} we get an $\alpha_1 + \alpha_2 = 10$-approximate solution for $\cI$ via \Cref{lma:approx}, and we are done.
\begin{proposition}\label{clm:case1}
If $y(L_1) \leq k_1-2$, then there is an integral solution $T$ for $\cT$ with $\w(\cC(T)) \geq \m$.
\end{proposition}
\begin{proof}
	Since $y(L_1)\leq k_1 -2$, we see that there is a feasible solution to the slightly revised LP below.
	\begin{align*}
    	\max\sum_{v \in L_2}\w(v)Y(v):~ 
    	&\sum_{u \in L_1} y_u \leq k_1 - 2,~~
    	\sum_{u\in L_2} y_u \leq k_2,~~\\
    	&Y(v) := y_{p(v)} + y_v \leq 1, ~\forall v \in L_2
	\end{align*}
	Consider a basic feasible solution $\{y'_v: v\in L_1\cup L_2\}$ for this LP, and let $T_1 := \{v\in L_1: y'_v > 0\}$.
    By definition $y'(T_1) = y'(L_1)\leq k_1 - 2$. According to \Cref{lma:loose}, there are at most 2 loose variables in $y'$. So there are at most 2 fractional vertices in $T_1$. This implies $|T_1| \leq k_1$.
	Let $U$ be the set of leaves that are not covered by $T_1$, that is, $U := \{v \in L_2: p(v) \notin T_1\}$. Let $T_2$ be the top $k_2$ members of $U$ according to decreasing $\w$ order. We return $T = (T_1,T_2)$.
	
	We claim $T$ has value at least $\m$, that is, $\w(\cC(T)) \geq \m$. Note that $\w(\cC(T)) = \w(T_2) + \sum_{u\in T_1} \w(\leaf(u))$.
	By the greedy choice of $T_2$, $\w(T_2) \geq \sum_{v \in U} \w(v)y'_v$. Since $y'_{p(v)} = 0$ for any $v \in U$, we have 
	$\w(T_2) \geq \sum_{v \in U} \w(v)y'_v = \sum_{v \in U} \w(v)Y'(v)$. 
	Furthermore, by definition, $\sum_{u\in T_1} \w(\leaf(u)) = \sum_{v\in L_2\setminus U} \w(v)$ which in turn is at least $\sum_{v\in L_2\setminus U} \w(v)Y'(v)$.
	Adding up proves the claim as the objective value is at least $\m$.
	\[
	\w(\cC(T)) \geq \sum_{v \in U} \w(v)Y'(v) + \sum_{v\in L_2\setminus U} \w(v)Y'_v = \sum_{v\in L_2} \w(v)Y'(v) \geq \m.
	\]
\end{proof}
\noindent
\textbf{Case II, $y(L_1) > k_1 - 2$.}
In this case, we either get an 8-approximation or prove that the following is a valid inequality which will serve as the separating hyperplane (recall $\pcov_1(L_1) = y(L_1)$).
\begin{equation}
    \cov_1(L_1) \leq k_1-2\label{eq:hypthm1:cov1}.
\end{equation}
To do so, we need the following proposition which formalizes the idea stated in~\Cref{sec:ouridea} that in case II, we can enumerate over $O(|X|)$ many {\em well-separated instances}. 
\begin{proposition}\label{clm:case2}
	Let $(\cov_1,\cov_2) \in \CovP$ be fractional coverages and suppose there is a subset $Y \subseteq X$ with $d(u,v) > 8r_1$ for all $u,v \in Y$. Then either $\cov_1(Y) \leq k_1 - 2$, or at least one of the following \septrkco instances are feasible
	\begin{align*}
	\cI_{\emptyset} &:= ((X,d),(2r_1,r_2),(k_1,k_2),Y,\m)&\\
	\cI_q &:= ((X\backslash B(q,r_1) ,d),(2r_1,r_2),(k_1-1,k_2),Y,\m-|B(q,r_1)|) & \forall q\in X:\conf{\\
	&&}d(q,Y) > r_1.
	\end{align*}
\end{proposition}
\full{Before proving the above proposition, let us use it}\conf{The proof is left to the full version of the paper. Now we use the above proposition} to complete the proof of~\Cref{thm:mainthm}. We let $Y := L_1$, and obtain the instances $\cI_{\emptyset}$ and $\cI_q$'s as mentioned in the proposition. We apply the algorithm in~\Cref{thm:approxsep} on each of them.
If any of them returns a solution, then we have an $8$-approximation. More precisely, if $\cI_{\emptyset}$ is feasible, \Cref{thm:approxsep} gives a 4-approximation for it which is indeed an 8-approximation for $\cI$ (the extra factor 2 is because $\cI_{\emptyset}$ uses $2r_1$ as its largest radius). If $\cI_q$ is feasible for some $q \in X$ and \Cref{thm:approxsep} gives us a 4-approximate solution $S' = (S'_1,S'_2)$ for it and $S = (S'_1\cup \{q\},S'_2)$ is an 8-approximation for $\cI$.
If none of them are feasible, then we see that $\cov_1(L_1)\leq k_1 - 2$ indeed serves as a separating hyperplane between $\pcov$ and $\CovP$. This ends the proof of~\Cref{thm:mainthm}.

\full{
\begin{proof}[Proof of \Cref{clm:case2}]
	Let us assume $\cov_1(Y) > k_1-2$, and prove that one of the proposed \septrkco instances are feasible.
	First of all, note that the described \septrkco instances indeed satisfy the definition:  $Y$ is separated enough for radius $2r_1$ and by definition of $q$, $Y \subseteq (X\backslash B(q,r_1))$.
	
	Suppose, for the sake of contradiction, none of the described \septrkco instances are feasible. Since $\cov_1(Y) > k_1-2$ and $\cov\in \CovP$, there has to be some $S  = (S_1,S_2) \in \cF$ such that $S_1$ covers {strictly} more than $k_1-2$ points in $Y$. Take any such $S$. There are two types of centers in $S_1$, the ones that do contribute to $\cov_1(Y)$, and the ones that do not. The former is $A := \{ f \in S_1: d(f,Y) \leq r_1\}$ and the latter is $B := \{ f \in S_1: d(f,Y) > r_1\}$. By our assumption of $\cov_1(Y) > k_1-2$ and the fact that $Y$ points are more than $2r_1$ apart, $|A| > k_1-2$. This leaves us with $|B| \leq 1$. 
	
	Let $\cI_B$ be the \septrkco instance corresponding to this $B$ (i.e. $\cI_B = \cI_{\emptyset}$ if $B = \emptyset$ and $\cI_B = \cI_q$ if $B = \{q\}$). We construct a feasible solution $S'$ for $\cI_B$ which contradicts our assumption. By definition of $\cI_B$, $S'$ only needs to cover as many points as $(A,S_2)$ covers in $\cI$. That is, $\m$ points if $B = \emptyset$ and $\m - |B(q,r_1)|$ points if $B = \{q\}$. Setting $S'_2 := S_2$ and $S'_1 := \{f \in Y: d(f,A) \leq 1\}$ does the trick: $S'$ balls with radii  $2r_1$ and $r_2$, cover all the elements in $X$ that are covered by $(A,S_2)$ in $\cI$. Also note that $S'$ satisfies the budget constraints.
\end{proof}}

\conf{\bibliographystyle{splncs04}}
\full{\bibliographystyle{alpha}}
\bibliography{refs}
\end{document}